\documentclass[12pt]{amsart}
\usepackage{amsmath,amssymb,amsthm}
\usepackage{mathtools}
\usepackage{stmaryrd} 
\usepackage[inline]{enumitem}
\usepackage{graphics} 
\usepackage{setspace}
\usepackage{tikz}
\usetikzlibrary{arrows}
\usepackage{wasysym}
\usepackage[makeroom]{cancel}
\usepackage{fullpage}
\usepackage{comment}
\usepackage{caption}
\usepackage{subcaption}
\usepackage{mathtools}
\usepackage{url}
\allowdisplaybreaks

\theoremstyle{plain}
\newtheorem{theorem}{Theorem}[section]
\newtheorem{proposition}[theorem]{Proposition}

\newtheorem{lemma}[theorem]{Lemma}

\newtheorem*{claim*}{Claim}
\theoremstyle{definition}
\newtheorem{remark}[theorem]{Remark}
\newtheorem{definition}[theorem]{Definition}

\usepackage[utf8]{inputenc}

\DeclareUnicodeCharacter{03B1}{\ensuremath{\alpha}}
\DeclareUnicodeCharacter{03B9}{\ensuremath{\iota}}
\DeclareUnicodeCharacter{03B2}{\ensuremath{\beta}}
\DeclareUnicodeCharacter{03BA}{\ensuremath{\kappa}}
\DeclareUnicodeCharacter{03F0}{\ensuremath{\varkappa}}
\DeclareUnicodeCharacter{03C3}{\ensuremath{\sigma}}
\DeclareUnicodeCharacter{03B3}{\ensuremath{\gamma}}
\DeclareUnicodeCharacter{03BB}{\ensuremath{\lambda}}
\DeclareUnicodeCharacter{03B4}{\ensuremath{\delta}}
\DeclareUnicodeCharacter{03BC}{\ensuremath{\mu}} 
\DeclareUnicodeCharacter{00B5}{\ensuremath{\mu}} 
\DeclareUnicodeCharacter{03C4}{\ensuremath{\tau}}
\DeclareUnicodeCharacter{03BD}{\ensuremath{\nu}}
\DeclareUnicodeCharacter{03C5}{\ensuremath{\upsilon}}
\DeclareUnicodeCharacter{03F5}{\ensuremath{\epsilon}}
\DeclareUnicodeCharacter{03B5}{\ensuremath{\varepsilon}}
\DeclareUnicodeCharacter{03BE}{\ensuremath{\xi}}
\DeclareUnicodeCharacter{03B6}{\ensuremath{\zeta}}
\DeclareUnicodeCharacter{03D5}{\ensuremath{\phi}}
\DeclareUnicodeCharacter{03C6}{\ensuremath{\varphi}}
\DeclareUnicodeCharacter{03B7}{\ensuremath{\eta}}
\DeclareUnicodeCharacter{03C0}{\ensuremath{\pi}}
\DeclareUnicodeCharacter{03D6}{\ensuremath{\varpi}}
\DeclareUnicodeCharacter{03C7}{\ensuremath{\chi}}
\DeclareUnicodeCharacter{03B8}{\ensuremath{\theta}}
\DeclareUnicodeCharacter{03C8}{\ensuremath{\psi}}
\DeclareUnicodeCharacter{03D1}{\ensuremath{\vartheta}}
\DeclareUnicodeCharacter{03C1}{\ensuremath{\rho}}
\DeclareUnicodeCharacter{03F1}{\ensuremath{\varrho}}
\DeclareUnicodeCharacter{03C9}{\ensuremath{\omega}}
\DeclareUnicodeCharacter{0393}{\ensuremath{\Gamma}}
\DeclareUnicodeCharacter{039E}{\ensuremath{\Xi}}
\DeclareUnicodeCharacter{03A6}{\ensuremath{\Phi}}
\DeclareUnicodeCharacter{0394}{\ensuremath{\Delta}}
\DeclareUnicodeCharacter{03A0}{\ensuremath{\Pi}}
\DeclareUnicodeCharacter{03A8}{\ensuremath{\Psi}}
\DeclareUnicodeCharacter{0398}{\ensuremath{\Theta}}
\DeclareUnicodeCharacter{03A3}{\ensuremath{\Sigma}}
\DeclareUnicodeCharacter{03A9}{\ensuremath{\Omega}}
\DeclareUnicodeCharacter{039B}{\ensuremath{\Lambda}}

\DeclareUnicodeCharacter{2207}{\ensuremath{\nabla}}
\DeclareUnicodeCharacter{2202}{\ensuremath{\partial}}
\DeclareUnicodeCharacter{222C}{\ensuremath{\iint}}
\DeclareUnicodeCharacter{222D}{\ensuremath{\iiint}}
\DeclareUnicodeCharacter{2A0C}{\ensuremath{\iiiint}}
\DeclareUnicodeCharacter{222E}{\ensuremath{\oint}}
\DeclareUnicodeCharacter{222F}{\ensuremath{\oiint}}
\DeclareUnicodeCharacter{2230}{\ensuremath{\oiiint}}

\DeclareUnicodeCharacter{2254}{:=}
\DeclareUnicodeCharacter{2255}{=:}
\DeclareUnicodeCharacter{2203}{\ensuremath{\exists}}
\DeclareUnicodeCharacter{2200}{\ensuremath{\forall}}

\DeclareUnicodeCharacter{22A4}{\top}                    
\DeclareUnicodeCharacter{22A5}{\bot}                    
\DeclareUnicodeCharacter{2135}{\ensuremath{\aleph}}
\DeclareUnicodeCharacter{2211}{\ensuremath{\sum}}
\DeclareUnicodeCharacter{222B}{\ensuremath{\int}}
\DeclareUnicodeCharacter{220F}{\ensuremath{\prod}}
\DeclareUnicodeCharacter{2261}{\ensuremath{\equiv}}
\DeclareUnicodeCharacter{2020}{\ensuremath{\dagger}} 
\DeclareUnicodeCharacter{00AC}{\ensuremath{\neg}}
\DeclareUnicodeCharacter{33D1}{\ensuremath{\ln}}
\DeclareUnicodeCharacter{33D2}{\ensuremath{\log}}
\DeclareUnicodeCharacter{221D}{\propto} 
\DeclareUnicodeCharacter{211C}{\ensuremath{\Re}}
\DeclareUnicodeCharacter{2111}{\ensuremath{\Im}}
\DeclareUnicodeCharacter{00A0}{~}
\DeclareUnicodeCharacter{202F}{\,}
\DeclareUnicodeCharacter{230A}{\ensuremath{\lfloor}}
\DeclareUnicodeCharacter{230B}{\ensuremath{\rfloor}}
\DeclareUnicodeCharacter{221A}{\ensuremath{\sqrt}}
\DeclareUnicodeCharacter{221B}{\ensuremath{\sqrt[3]}}
\DeclareUnicodeCharacter{221C}{\ensuremath{\sqrt[4]}}
\DeclareUnicodeCharacter{00D7}{\ensuremath{\times}}
\DeclareUnicodeCharacter{00F7}{\ensuremath{\div}}
\DeclareUnicodeCharacter{00B1}{\ensuremath{\pm}}
\DeclareUnicodeCharacter{2213}{\ensuremath{\mp}}
\DeclareUnicodeCharacter{2215}{\ensuremath{/}}
\DeclareUnicodeCharacter{2212}{\ensuremath{-}}
\DeclareUnicodeCharacter{221E}{\ensuremath{\infty}}
\DeclareUnicodeCharacter{2228}{\ensuremath{\vee}}
\DeclareUnicodeCharacter{2227}{\ensuremath{\wedge}}

\DeclareUnicodeCharacter{22C1}{\ensuremath{\bigvee}}
\DeclareUnicodeCharacter{22C0}{\ensuremath{\bigwedge}}
\DeclareUnicodeCharacter{22C3}{\ensuremath{\bigcup}}
\DeclareUnicodeCharacter{22C2}{\ensuremath{\bigcap}}
\DeclareUnicodeCharacter{2A00}{\ensuremath{\bigodot}}
\DeclareUnicodeCharacter{2A01}{\ensuremath{\bigoplus}}
\DeclareUnicodeCharacter{2A02}{\ensuremath{\bigotimes}}

\DeclareUnicodeCharacter{2260}{\ensuremath{\neq}}
\DeclareUnicodeCharacter{2248}{\ensuremath{\approx}}
\DeclareUnicodeCharacter{2264}{\ensuremath{\leq}}
\DeclareUnicodeCharacter{2265}{\ensuremath{\geq}}
\DeclareUnicodeCharacter{226A}{\ensuremath{\ll}}
\DeclareUnicodeCharacter{226B}{\ensuremath{\gg}}
\DeclareUnicodeCharacter{226D}{\not\equiv}              
\DeclareUnicodeCharacter{2270}{\not\le}                 
\DeclareUnicodeCharacter{2271}{\not\ge}                 
\DeclareUnicodeCharacter{2272}{\lesssim}                
\DeclareUnicodeCharacter{2273}{\gtrsim}                 
\DeclareUnicodeCharacter{227A}{\prec}                   
\DeclareUnicodeCharacter{2282}{\subset}                 
\DeclareUnicodeCharacter{2286}{\subseteq}               

\DeclareUnicodeCharacter{220B}{\ensuremath{\ni}}
\DeclareUnicodeCharacter{2205}{\ensuremath{\emptyset}}
\DeclareUnicodeCharacter{2208}{\ensuremath{\in}}
\DeclareUnicodeCharacter{2282}{\ensuremath{\subset}}
\DeclareUnicodeCharacter{2283}{\ensuremath{\supset}}
\DeclareUnicodeCharacter{2286}{\ensuremath{\subseteq}}
\DeclareUnicodeCharacter{2287}{\ensuremath{\supseteq}}
\DeclareUnicodeCharacter{2229}{\ensuremath{\cap}}
\DeclareUnicodeCharacter{222A}{\ensuremath{\cup}}
\DeclareUnicodeCharacter{2288}{\ensuremath{\nsubseteq}} 

\DeclareUnicodeCharacter{2026}{\ifmmode\ldots\else\textellipsis\fi} 
\DeclareUnicodeCharacter{22C5}{\ensuremath{\cdot}}
\DeclareUnicodeCharacter{22EE}{\ensuremath{\vdots}}                  

\DeclareUnicodeCharacter{21D2}{\ensuremath{\Rightarrow}}
\DeclareUnicodeCharacter{21D0}{\ensuremath{\Leftarrow}}
\DeclareUnicodeCharacter{21D4}{\ensuremath{\Leftrightarrow}}
\DeclareUnicodeCharacter{2192}{\ensuremath{\to}}
\DeclareUnicodeCharacter{2190}{\ensuremath{\gets}}
\DeclareUnicodeCharacter{21A6}{\ensuremath{\mapsto}}

\DeclareUnicodeCharacter{20AC}{\EUR}

\DeclareUnicodeCharacter{2070}{\ensuremath{^0}}
\DeclareUnicodeCharacter{00B9}{\ensuremath{^1}}
\DeclareUnicodeCharacter{00B2}{\ensuremath{^2}}
\DeclareUnicodeCharacter{00B3}{\ensuremath{^3}}
\DeclareUnicodeCharacter{2074}{\ensuremath{^4}}
\DeclareUnicodeCharacter{2075}{\ensuremath{^5}}
\DeclareUnicodeCharacter{2076}{\ensuremath{^6}}
\DeclareUnicodeCharacter{2077}{\ensuremath{^7}}
\DeclareUnicodeCharacter{2078}{\ensuremath{^8}}
\DeclareUnicodeCharacter{2079}{\ensuremath{^9}}
\DeclareUnicodeCharacter{207A}{\ensuremath{^+}}
\DeclareUnicodeCharacter{207B}{\ensuremath{^-}}
\DeclareUnicodeCharacter{207C}{\ensuremath{^=}}
\DeclareUnicodeCharacter{207D}{\ensuremath{^(}}
\DeclareUnicodeCharacter{207E}{\ensuremath{^)}}
\DeclareUnicodeCharacter{2080}{\ensuremath{_0}}
\DeclareUnicodeCharacter{2081}{\ensuremath{_1}}
\DeclareUnicodeCharacter{2082}{\ensuremath{_2}}
\DeclareUnicodeCharacter{2083}{\ensuremath{_3}}
\DeclareUnicodeCharacter{2084}{\ensuremath{_4}}
\DeclareUnicodeCharacter{2085}{\ensuremath{_5}}
\DeclareUnicodeCharacter{2086}{\ensuremath{_6}}
\DeclareUnicodeCharacter{2087}{\ensuremath{_7}}
\DeclareUnicodeCharacter{2088}{\ensuremath{_8}}
\DeclareUnicodeCharacter{2089}{\ensuremath{_9}}
\DeclareUnicodeCharacter{208A}{\ensuremath{_+}}
\DeclareUnicodeCharacter{208B}{\ensuremath{_-}}
\DeclareUnicodeCharacter{208C}{\ensuremath{_=}}
\DeclareUnicodeCharacter{208D}{\ensuremath{_(}}
\DeclareUnicodeCharacter{208E}{\ensuremath{_)}}

\DeclareUnicodeCharacter{2096}{_k}
\DeclareUnicodeCharacter{2097}{_l}
\DeclareUnicodeCharacter{2098}{_m}
\DeclareUnicodeCharacter{2099}{_n}
\DeclareUnicodeCharacter{2099}{_n}
\DeclareUnicodeCharacter{209B}{_s}
\DeclareUnicodeCharacter{209C}{_t}
\DeclareUnicodeCharacter{1D64}{_u}

\DeclareUnicodeCharacter{1D538}{\ensuremath{\mathbb{A}}}
\DeclareUnicodeCharacter{1D539}{\ensuremath{\mathbb{B}}}
\DeclareUnicodeCharacter{02102}{\ensuremath{\mathbb{C}}}
\DeclareUnicodeCharacter{1D53B}{\ensuremath{\mathbb{D}}}
\DeclareUnicodeCharacter{1D53C}{\ensuremath{\mathbb{E}}}
\DeclareUnicodeCharacter{1D53D}{\ensuremath{\mathbb{F}}}
\DeclareUnicodeCharacter{1D53E}{\ensuremath{\mathbb{G}}}
\DeclareUnicodeCharacter{0210D}{\ensuremath{\mathbb{H}}}
\DeclareUnicodeCharacter{1D540}{\ensuremath{\mathbb{I}}}
\DeclareUnicodeCharacter{1D541}{\ensuremath{\mathbb{J}}}
\DeclareUnicodeCharacter{1D542}{\ensuremath{\mathbb{K}}}
\DeclareUnicodeCharacter{1D543}{\ensuremath{\mathbb{L}}}
\DeclareUnicodeCharacter{1D544}{\ensuremath{\mathbb{M}}}
\DeclareUnicodeCharacter{02115}{\ensuremath{\mathbb{N}}}
\DeclareUnicodeCharacter{1D546}{\ensuremath{\mathbb{O}}}
\DeclareUnicodeCharacter{02119}{\ensuremath{\mathbb{P}}}
\DeclareUnicodeCharacter{0211A}{\ensuremath{\mathbb{Q}}}
\DeclareUnicodeCharacter{0211D}{\ensuremath{\mathbb{R}}}
\DeclareUnicodeCharacter{1D54A}{\ensuremath{\mathbb{S}}}
\DeclareUnicodeCharacter{1D54B}{\ensuremath{\mathbb{T}}}
\DeclareUnicodeCharacter{1D54C}{\ensuremath{\mathbb{U}}}
\DeclareUnicodeCharacter{1D54D}{\ensuremath{\mathbb{V}}}
\DeclareUnicodeCharacter{1D54E}{\ensuremath{\mathbb{W}}}
\DeclareUnicodeCharacter{1D54F}{\ensuremath{\mathbb{X}}}
\DeclareUnicodeCharacter{1D550}{\ensuremath{\mathbb{Y}}}
\DeclareUnicodeCharacter{02124}{\ensuremath{\mathbb{Z}}}
\DeclareUnicodeCharacter{1D552}{\ensuremath{\mathbb{a}}}
\DeclareUnicodeCharacter{1D553}{\ensuremath{\mathbb{b}}}
\DeclareUnicodeCharacter{1D554}{\ensuremath{\mathbb{c}}}
\DeclareUnicodeCharacter{1D555}{\ensuremath{\mathbb{d}}}
\DeclareUnicodeCharacter{1D556}{\ensuremath{\mathbb{e}}}
\DeclareUnicodeCharacter{1D557}{\ensuremath{\mathbb{f}}}
\DeclareUnicodeCharacter{1D558}{\ensuremath{\mathbb{g}}}
\DeclareUnicodeCharacter{1D559}{\ensuremath{\mathbb{h}}}
\DeclareUnicodeCharacter{1D55A}{\ensuremath{\mathbb{i}}}
\DeclareUnicodeCharacter{1D55B}{\ensuremath{\mathbb{j}}}
\DeclareUnicodeCharacter{1D55C}{\ensuremath{\mathbb{k}}}
\DeclareUnicodeCharacter{1D55D}{\ensuremath{\mathbb{l}}}
\DeclareUnicodeCharacter{1D55E}{\ensuremath{\mathbb{m}}}
\DeclareUnicodeCharacter{1D55F}{\ensuremath{\mathbb{n}}}
\DeclareUnicodeCharacter{1D560}{\ensuremath{\mathbb{o}}}
\DeclareUnicodeCharacter{1D561}{\ensuremath{\mathbb{p}}}
\DeclareUnicodeCharacter{1D562}{\ensuremath{\mathbb{q}}}
\DeclareUnicodeCharacter{1D563}{\ensuremath{\mathbb{r}}}
\DeclareUnicodeCharacter{1D564}{\ensuremath{\mathbb{s}}}
\DeclareUnicodeCharacter{1D565}{\ensuremath{\mathbb{t}}}
\DeclareUnicodeCharacter{1D566}{\ensuremath{\mathbb{u}}}
\DeclareUnicodeCharacter{1D567}{\ensuremath{\mathbb{v}}}
\DeclareUnicodeCharacter{1D568}{\ensuremath{\mathbb{w}}}
\DeclareUnicodeCharacter{1D569}{\ensuremath{\mathbb{x}}}
\DeclareUnicodeCharacter{1D56A}{\ensuremath{\mathbb{y}}}
\DeclareUnicodeCharacter{1D56B}{\ensuremath{\mathbb{z}}}
\DeclareUnicodeCharacter{1D7D8}{\ensuremath{\mathbb{0}}}
\DeclareUnicodeCharacter{1D7D9}{\ensuremath{\mathbb{1}}}
\DeclareUnicodeCharacter{1D7DA}{\ensuremath{\mathbb{2}}}
\DeclareUnicodeCharacter{1D7DB}{\ensuremath{\mathbb{3}}}
\DeclareUnicodeCharacter{1D7DC}{\ensuremath{\mathbb{4}}}
\DeclareUnicodeCharacter{1D7DD}{\ensuremath{\mathbb{5}}}
\DeclareUnicodeCharacter{1D7DE}{\ensuremath{\mathbb{6}}}
\DeclareUnicodeCharacter{1D7DF}{\ensuremath{\mathbb{7}}}
\DeclareUnicodeCharacter{1D7E0}{\ensuremath{\mathbb{8}}}
\DeclareUnicodeCharacter{1D7E1}{\ensuremath{\mathbb{9}}}

\DeclareUnicodeCharacter{1D504}{\ensuremath{\mathfrak{A}}}
\DeclareUnicodeCharacter{1D505}{\ensuremath{\mathfrak{B}}}
\DeclareUnicodeCharacter{1D507}{\ensuremath{\mathfrak{D}}}
\DeclareUnicodeCharacter{1D508}{\ensuremath{\mathfrak{E}}}
\DeclareUnicodeCharacter{1D509}{\ensuremath{\mathfrak{F}}}
\DeclareUnicodeCharacter{1D50A}{\ensuremath{\mathfrak{G}}}
\DeclareUnicodeCharacter{1D50D}{\ensuremath{\mathfrak{J}}}
\DeclareUnicodeCharacter{1D50E}{\ensuremath{\mathfrak{K}}}
\DeclareUnicodeCharacter{1D50F}{\ensuremath{\mathfrak{L}}}
\DeclareUnicodeCharacter{1D510}{\ensuremath{\mathfrak{M}}}
\DeclareUnicodeCharacter{1D511}{\ensuremath{\mathfrak{N}}}
\DeclareUnicodeCharacter{1D512}{\ensuremath{\mathfrak{O}}}
\DeclareUnicodeCharacter{1D513}{\ensuremath{\mathfrak{P}}}
\DeclareUnicodeCharacter{1D514}{\ensuremath{\mathfrak{Q}}}
\DeclareUnicodeCharacter{1D516}{\ensuremath{\mathfrak{S}}}
\DeclareUnicodeCharacter{1D517}{\ensuremath{\mathfrak{T}}}
\DeclareUnicodeCharacter{1D518}{\ensuremath{\mathfrak{U}}}
\DeclareUnicodeCharacter{1D519}{\ensuremath{\mathfrak{V}}}
\DeclareUnicodeCharacter{1D51A}{\ensuremath{\mathfrak{W}}}
\DeclareUnicodeCharacter{1D51B}{\ensuremath{\mathfrak{X}}}
\DeclareUnicodeCharacter{1D51C}{\ensuremath{\mathfrak{Y}}}

\DeclareUnicodeCharacter{1D4D0}{\ensuremath{\mathcal{A}}}
\DeclareUnicodeCharacter{1D4D1}{\ensuremath{\mathcal{B}}}
\DeclareUnicodeCharacter{1D4D2}{\ensuremath{\mathcal{C}}}
\DeclareUnicodeCharacter{1D4D3}{\ensuremath{\mathcal{D}}}
\DeclareUnicodeCharacter{1D4D4}{\ensuremath{\mathcal{E}}}
\DeclareUnicodeCharacter{1D4D5}{\ensuremath{\mathcal{F}}}
\DeclareUnicodeCharacter{1D4D6}{\ensuremath{\mathcal{G}}}
\DeclareUnicodeCharacter{1D4D7}{\ensuremath{\mathcal{H}}}
\DeclareUnicodeCharacter{1D4D8}{\ensuremath{\mathcal{I}}}
\DeclareUnicodeCharacter{1D4D9}{\ensuremath{\mathcal{J}}}
\DeclareUnicodeCharacter{1D4DA}{\ensuremath{\mathcal{K}}}
\DeclareUnicodeCharacter{1D4DB}{\ensuremath{\mathcal{L}}}
\DeclareUnicodeCharacter{1D4DC}{\ensuremath{\mathcal{M}}}
\DeclareUnicodeCharacter{1D4DD}{\ensuremath{\mathcal{N}}}
\DeclareUnicodeCharacter{1D4DE}{\ensuremath{\mathcal{O}}}
\DeclareUnicodeCharacter{1D4DF}{\ensuremath{\mathcal{P}}}
\DeclareUnicodeCharacter{1D4E0}{\ensuremath{\mathcal{Q}}}
\DeclareUnicodeCharacter{1D4E1}{\ensuremath{\mathcal{R}}}
\DeclareUnicodeCharacter{1D4E2}{\ensuremath{\mathcal{S}}}
\DeclareUnicodeCharacter{1D4E3}{\ensuremath{\mathcal{T}}}
\DeclareUnicodeCharacter{1D4E4}{\ensuremath{\mathcal{U}}}
\DeclareUnicodeCharacter{1D4E5}{\ensuremath{\mathcal{V}}}
\DeclareUnicodeCharacter{1D4E6}{\ensuremath{\mathcal{W}}}
\DeclareUnicodeCharacter{1D4E7}{\ensuremath{\mathcal{X}}}
\DeclareUnicodeCharacter{1D4E8}{\ensuremath{\mathcal{Y}}}
\DeclareUnicodeCharacter{1D4E9}{\ensuremath{\mathcal{Z}}}

\DeclareUnicodeCharacter{1D43}{^a}
\DeclareUnicodeCharacter{1D47}{^b}
\DeclareUnicodeCharacter{1D9C}{^c}
\DeclareUnicodeCharacter{1D48}{^d}
\DeclareUnicodeCharacter{1D49}{^e}
\DeclareUnicodeCharacter{1DA0}{^f}
\DeclareUnicodeCharacter{1D4D}{^g}
\DeclareUnicodeCharacter{02B0}{^h}
\DeclareUnicodeCharacter{2071}{^i}
\DeclareUnicodeCharacter{02B2}{^j}
\DeclareUnicodeCharacter{1D4F}{^k}
\DeclareUnicodeCharacter{02E1}{^l}
\DeclareUnicodeCharacter{1D50}{^m}
\DeclareUnicodeCharacter{207F}{^n}
\DeclareUnicodeCharacter{1D52}{^o}
\DeclareUnicodeCharacter{1D56}{^p}
\DeclareUnicodeCharacter{02B3}{^r}
\DeclareUnicodeCharacter{02E2}{^s}
\DeclareUnicodeCharacter{1D57}{^t}
\DeclareUnicodeCharacter{1D58}{^u}
\DeclareUnicodeCharacter{1D5B}{^v}
\DeclareUnicodeCharacter{02B7}{^w}
\DeclareUnicodeCharacter{02E3}{^x}
\DeclareUnicodeCharacter{02B8}{^y}
\DeclareUnicodeCharacter{1DBB}{^z}
\DeclareUnicodeCharacter{1D2C}{^A}
\DeclareUnicodeCharacter{1D2E}{^B}
\DeclareUnicodeCharacter{1D30}{^D}
\DeclareUnicodeCharacter{1D31}{^E}
\DeclareUnicodeCharacter{1D33}{^G}
\DeclareUnicodeCharacter{1D34}{^H}
\DeclareUnicodeCharacter{1D35}{^I}
\DeclareUnicodeCharacter{1D36}{^J}
\DeclareUnicodeCharacter{1D37}{^K}
\DeclareUnicodeCharacter{1D38}{^L}
\DeclareUnicodeCharacter{1D39}{^M}
\DeclareUnicodeCharacter{1D3A}{^N}
\DeclareUnicodeCharacter{1D3C}{^O}
\DeclareUnicodeCharacter{1D3E}{^P}
\DeclareUnicodeCharacter{1D3F}{^R}
\DeclareUnicodeCharacter{1D40}{^T}
\DeclareUnicodeCharacter{1D41}{^U}
\DeclareUnicodeCharacter{1D42}{^W}

\DeclareUnicodeCharacter{1D49E}{\mathcal{C}}            
\DeclareUnicodeCharacter{1D49F}{\mathcal{D}}            
\DeclareUnicodeCharacter{1D4BB}{\mathcal{f}}            
\DeclareUnicodeCharacter{1D4DB}{\mathcal{L}}            
\DeclareUnicodeCharacter{1D4E3}{\mathcal{T}}            

\DeclareUnicodeCharacter{2308}{\left\lceil}
\DeclareUnicodeCharacter{2309}{\right\rceil}
\DeclareUnicodeCharacter{1D62}{\ensuremath{_i}}
\DeclareUnicodeCharacter{00B7}{\ensuremath{\cdot}}
\DeclareUnicodeCharacter{223C}{\ensuremath{\sim}}
\DeclareUnicodeCharacter{00D5}{\ifmmode\tilde{O}\else\~{O}\fi}
\DeclareUnicodeCharacter{2245}{\ensuremath{\cong}}
\DeclareUnicodeCharacter{22EF}{\ensuremath{\cdots}}
\DeclareUnicodeCharacter{2C7C}{\ensuremath{_j}}
\DeclareUnicodeCharacter{209A}{\ensuremath{_p}}
\DeclareUnicodeCharacter{2297}{\ensuremath{\otimes}}
\DeclareUnicodeCharacter{21AA}{\xhookrightarrow{}}
\DeclareUnicodeCharacter{2016}{\Vert}
\DeclareUnicodeCharacter{2113}{\ell}

\DeclareUnicodeCharacter{2223}{\mid}
\DeclareUnicodeCharacter{2224}{\nmid}

\DeclareUnicodeCharacter{220B}{\ensuremath{\ni}}
\DeclareUnicodeCharacter{2209}{\not\in}

\DeclareMathOperator{\Z}{\mathbb{Z}}

\DeclareMathOperator{\E}{\mathbb{E}}

\DeclareMathOperator{\Var}{Var}

\newcommand{\norm}[1]{\left\lVert#1\right\rVert}

\newif\ifshowold 
\showoldfalse 
\title{The Polynomial Learning With Errors Problem and the Smearing Condition}

\author [L. Babinkostova] {L.\ Babinkostova $^1$}
\address{$^1$ Department of Mathematics, Boise State University, Boise, ID 83714.}

\author[A. Chin]{A.\ Chin $^2$}
\address{$^2$  Department of Mathematics, University of California, Berkeley, Berkeley, CA 94712.}

\author[A. Kirtland ]{A.\ Kirtland $^3$ }
\address{ $^3$  Department of Mathematics, Washington University in St. Louis, St. Louis, MO 63130.}

\author[V. Nazarchuk]{V. \ Nazarchuk $^4$}
\address{$^4$  Department of Mathematics, Yale University, New Haven, CT 06520.}

\author[E. Plotnick]{E. \ Plotnick $^5$}
\address{$^5$  Department of Mathematics, Harvard University, Cambridge, MA 02138.}

\thanks{Supported by the National Science Foundation under the grant number DMS-1659872.}
\thanks{$^{\S}$ Corresponding Author: liljanababinkostova@boisestate.edu}
\begin{document}

\begin{abstract}
As quantum computing advances rapidly, guaranteeing the security of cryptographic protocols resistant to quantum attacks is paramount. Some leading candidate cryptosystems use the Learning with Errors (LWE) problem, attractive for its simplicity and hardness guaranteed by reductions from hard computational lattice problems. Its algebraic variants, Ring-Learning with Errors (RLWE) and Polynomial Learning with Errors (PLWE), gain in efficiency over standard LWE, but their security remains to be thoroughly investigated. In this work, we consider the ``smearing'' condition, a condition for attacks on PLWE and RLWE 
introduced in \cite {provably weak}.
We expand upon some questions about smearing posed by Elias et al. in \cite{provably weak}  and show how smearing is related to the Coupon Collector's Problem. Furthermore, we develop some practical algorithms for calculating probabilities related to smearing. Finally, we present a smearing-based attack on PLWE, and demonstrate its effectiveness.
\end{abstract}
\maketitle
\section{Introduction}
Quantum computing promises to be a game-changing technology, as many problems that are considered intractable for conventional computers could be solved efficiently by harnessing properties of quantum physics to represent information. While quantum computing provides new methods to approach complex computing problems, it can also be used as a powerful tool to break existing cryptographic security. 

There are currently two groundbreaking quantum algorithms which break today's conventional cryptosystems. In 1994, Shor \cite{Shor} proposed an efficient polynomial-time quantum algorithm for solving the integer factorization and discrete log problems. Indeed, this algorithm breaks much of public key cryptography, as many widely-used public key cryptosystems rely on the difficulty of integer factorization and elliptic curve variants of the discrete logarithm problem, both of which have no known polynomial-time solution with conventional computing. 
In 1996, Grover \cite{G} proposed a quantum algorithm that provides a quadratic speed up over classical algorithms for searching a key space, weakening the security of symmetric key cryptosystems which rely on the difficulty of guessing a random shared key. 

To address this issue, in 2016 the National Institute of Standard and Technology \cite{NIST} announced the need to replace cryptosystems and standards based on vulnerable problems with post-quantum cryptography alternatives. 
A promising avenue in post-quantum cryptography is \textit{lattice-based cryptography}, cryptography based on well-studied computational problems on lattices which have no known efficient solution with either classical or quantum computing. Some lattice-based cryptography relies on the Learning With Errors (LWE) problem, introduced by O. Regev \cite{LP1} in 2005, which exploits the difficulty of solving a ``noisy" linear system modulo a known integer. Regev also proved a reduction from worst-case computational lattice problems to LWE, affirming its difficulty and making LWE a strong candidate to base cryptographic systems.

The basic search LWE problem takes the form of a linear system hiding a secret integer vector $s$, with integer coefficient vector $a$, and integer error vector $e$, modulo some integer $q$:
\begin{align*}
    a_1s_1 + e_1 &\equiv c_1\mod{q}\\
    a_2s_2 + e_2 &\equiv c_2\mod{q}\\
    a_3s_3 + e_3 &\equiv c_3\mod{q}\\
    &\vdots
\end{align*}
While Gaussian elimination makes this system easy to solve with known $a_i$, $e_i$, and $c_i$, the introduction of the unknown noise $e_i$ makes an easy linear system extremely difficult to solve. Even with small noise, the traditional process of Gaussian elimination magnifies noise to the point of rendering the modular linear system unsolvable.

The decision LWE problem is to distinguish with non-negligible advantage between a uniform distribution and a distribution over the noisy inner products $(a,as+e)$ (where $a$ is sampled uniformly at random). Since its introduction, the conjectured hardness of LWE \cite{LP1} has already been used as a building block for many cryptographic applications: in efficient signature schemes \cite{W}, fully-homomorphic encryption schemes \cite{BV}, pseudo-random functions \cite{B}, and protocols for secure multi-party computation \cite{DPA}, and it also validates the hardness of the NTRU cryptosystem \cite{HPS}. 

The ``algebraically structured" variants, called Ring LWE (RLWE) \cite{MR}, Polynomial LWE (PLWE) \cite{LP}, and Module LWE \cite{A2} (drawing values from any ring of integers, polynomial rings, and modules in place of the set of integers respectively), offer more succinct representations of information. 


While the hardness of RLWE relies on the conjectured hardness of computational lattice problems over a restricted set of lattices (called \textit{ideal} lattices) \cite{on ideal lattices}, its construction has inherent algebraic structure, which could make it vulnerable to algebraic attacks. In this paper we consider attacks against the PLWE problem.


We analyze a condition for attacks against the PLWE problem called the \emph{smearing condition}, which was introduced by Elias, Lauter, Ozman, and Stange in \cite{RLWE for NT}. We demonstrate the parallels between the smearing condition and the Coupon Collector's Problem and develop recursive methods for computing the probability of smearing. We also present a new attack on the PLWE decision problem using the smearing condition.

The paper is organized as follows: Section 2 summarizes relevant background related to the RLWE problems, Section 3 focuses on the smearing condition and gives an overview of related work, Section 4 provides methods of calculating smearing probabilities for both uniform and non-uniform distributions, and in Section 5 we provide a smearing-based attack on the PLWE problem. 

\section{Preliminaries} 
\subsection{Lattices and Gaussians} A \emph{lattice} is a discrete additive subgroup of a vector space $\mathbb{V}$. If $\mathbb{V}$ has dimension $n$ a lattice $L$ can be viewed as the set of all integer linear combinations of a set of linearly independent vectors $B = \{b_1, \cdots, b_k\}$ for some $k\leq n$, written $\mathcal{L} = L(B) = \sum_{i=1}^k: z_i\boldsymbol{b}_i, z_i\in \mathbb{Z}\}$. If $k = n$ we call the lattice full-rank, and we will only consider lattices of full-rank. We can extend this notion of lattices to matrix spaces by stacking the columns of a matrix. 
We recall the following standard definitions of lattices and Gaussians.

\begin{definition} Given a lattice $\mathcal{L}$ in a space $V$ endowed with a metric $\| \cdot \|$, the minimum distance of $\mathcal{L}$ is defined as $\lambda_1(\Lambda) = min_{0\neq \boldsymbol{v}\in\mathcal{L}}\Vert v \|$. Similarly, $\lambda_n$ is the minimum length of a set of $n$ linearly independent vectors, where the length of a set of vectors $\{\boldsymbol{x}_1, \cdots, \boldsymbol{x}_k\}$ is defined as $max_i \Vert \boldsymbol{x}_i \|$.
\end{definition}

\begin{definition}
Given a lattice $\mathcal{L} \subseteq \mathbb{V}$, where $V$ is endowed with an inner product $\langle\cdot,\cdot\rangle$, the dual lattice $\mathcal{L}^*$ is defined $\mathcal{L}^* =\{ \boldsymbol{v}\in \mathbb{V}: \langle \Lambda,\boldsymbol{v} \rangle \subset\mathbb{Z}\}$.
\end {definition}

For a vector space $\mathbb{V}$ with norm $\|\cdot\|$ and $σ > 0$,
we define the Gaussian function $\rho_{σ}: ℝ^n\rightarrow(0,1]$ by \[\rho_{σ}(x) = e^{-\frac{π\norm{x}^2}{2σ^2}}.\]



The \textit{Gaussian distribution} (normal distribution) with parameter $σ >0$ has a continuous probability density function
\[f(x) = \frac{1}{σ \sqrt{2π}}e^{-\frac{π\norm{x}^2}{2σ^2}}.\]

When sampling a Gaussian over a lattice $\mathcal{L}$ we will use the discrete form of the Gaussian distribution. This Gaussian distribution is discretized as follows.

\begin{definition}
A \textit{discrete Gaussian distribution} $\mathcal{DG}_σ$ with parameter $σ >0$ over a lattice $\mathcal{L}\subset ℝ^n$ is a distribution in which each lattice point $λ$ is sampled with probability $P_{λ}$ proportional to $\rho_σ(λ)$.
\[P_{λ}≔ \frac{\rho_σ(λ)}{\rho_σ(\mathcal{L})},\rho_σ(\mathcal{L}) = ∑_{λ\in\mathcal{L}}ρ_r(λ)\]
\end{definition}

It is well-known that the sum of $n$ independent, normally-distributed random variables is normal. 

\begin{definition}[Informal]
A \textit{spherical Gaussian} distribution is a multivariate Gaussian distribution such that there are no interactions between the dimensions.
\end{definition}
This implies that we can simply select each coordinate from a Gaussian distribution. We use the Gaussian distribution as the error distribution in the Learning with Errors problem, discussed below.

\subsection{Learning with Errors Distributions}

Let $f(x)$ be a monic irreducible polynomial in $\mathbb{Z}[x]$ of degree $n$. We use the notation $\mathbb{R}_q$ to denote the polynomial ring $\mathbb{Z}_q[x]/(f(x))$. 


An instance of the RLWE distribution is given by a choice of number field $K$, secret $s$, prime $q$, and parameter $σ$ (for the error distribution). 

\begin{definition}[RLWE Distribution, \cite{RLWE for NT}]
Let $R=\mathcal{O}_K$ be the ring of integers for number field $K$. Define \[R_q=R/qR.\] Let $\mathcal{U}_{R_q}$ be the uniform distribution over $R_q$. Let the error distribution be $\mathcal{G}_{σ, R_q}$, a discrete Gaussian distribution over $R_q$. 
For some $s\in R_q$, $a\leftarrow \mathcal{U_{P_q}}$, $e\leftarrow \mathcal{G}_{σ, R_q}$, 
pairs of the form
$(a, c=a⋅ s + e)$
compose the \textit{RLWE distribution} $\mathcal{D}_{R, s, G_{σ}}$ over $R_q\times R_q$.
\end{definition}

The PLWE distribution is defined similarly; rather than the ring of integers of a number field, the distribution is defined over a polynomial ring. An instance of the PLWE distribution is now given by a choice of monic, irreducible polynomial $f(x)\in ℤ[x]$, secret $s$, prime $q$, and parameter $σ$ (for the error distribution). 

\begin{definition}[PLWE Distribution, \cite{RLWE for NT}]
Let $f(x)\in ℤ[x]$ be monic, irreducible of degree $n$. Assume that $f(x)$ splits over $ℤ_q≔ ℤ/qℤ$. Define
\[P\coloneqq ℤ[x]/(f(x)), P_q\coloneqq P/qP.\]

Let $G_{σ, P}$ be a discrete Gaussian distribution over $P$ spherical in the power basis of $P$ $(1,x,x^2,\ldots, x^{n-1})$. Let $\mathcal{U}_{P_q}$ be the uniform distribution over $P_q$. Let the error distribution be $\mathcal{G}_{σ, P_q}$ a discrete Gaussian distribution over $P_q$. 

For some $s\in P_q$, $a\leftarrow \mathcal{U}_{P_q}$, and $e\leftarrow \mathcal{G}_{σ, P_q}$, pairs of the form
\[(a, c=a⋅ s + e)\]
compose the \textit{PLWE distribution} $\mathcal{D}_{P, s, G_{σ}}$ over $P_q\times P_q$.
\end{definition}

The decision problems for RLWE and PLWE are analogous to Decision LWE: given the same number of (arbitrarily many) independent samples from two distributions, determine with non-negligible advantage that the set of samples follows the RLWE (PLWE) distribution $\mathcal{D}_{R,s,G_{σ}}$ ($\mathcal{D}_{P,s,G_{σ}}$) versus a uniform distribution over $R_q\times R_q$ ($P_q\times P_q$).

\section{Smearing Condition}
\subsection{Motivation and Related Work}
A common technique for breaking cryptographic schemes is to transfer the problem onto a smaller space, where looking for the secret key by brute force is feasible. In finding the secret $s\in P_q$ in a PLWE problem by brute force, the attacker would have to go through $q^n$ different possibilities, which is infeasible due to the sizes of $q$ and $n$. However, if the attacker can somehow transfer the PLWE problem onto a smaller field, like $\Z_q$, then brute force suddenly becomes feasible, and, if not much information is lost in this transformation, then a brute search on $\Z_q$ would help solve the original problem on $P_q$.

An example of this approach is the ``$\gamma=1$ attack" on Decision-PLWE, as presented in \cite{RLWE for NT}. Suppose that $f$ has a root at $\gamma=1$, i.e. $f(1)≡ 0 \mod q$. Expressing $e$ in the power basis, $e(x)=∑_{j=0}^{n-1} e_j x^j$, where $e_j \sim \mathcal{G}_{\sqrt{n}σ}$. Then, $e(1)=∑_{j=0}^{n-1} e_j \leftarrow \mathcal{G}_{\sqrt{n}σ}$. So, if samples follow the PLWE distribution, $e(1)$ can only take on a \emph{small range of values}.


Note that there are $q$ possibilities for the value of $s(1)$. So,  
\begin{enumerate}
    \item For all possible guesses $g\in \Z_q$ (where $g$ is a guess for $s(1)$ and for each sample $(a_i, c_i)$, compute $e_i'=c_i(1)-g ⋅ a_i(1)$.
    \begin{itemize}
    \item [(a)] Check and record if $e_i'$ is within $\mathcal{G}_{\sqrt{n}σ}$. \\ 
    Note: If the guess for $s(1)$ is correct, $e_i'$ will equal $e_i(1)$. If the guess for $s(1)$ is incorrect, or if $(a_i, c_i)$ are uniform to begin with, $e_i'$ will be uniform over $\Z_q$.
    \end{itemize}
    \item Make a decision about the sample distribution: 
    \begin{enumerate}
        \item If there is one $g$ for which all the $c_i(1)-g ⋅ a_i(1)$'s are within $\mathcal{G}_{\sqrt{n}σ}$, the $(a_i, c_i)$ are taken from the \textit{PLWE distribution} with $g=s(1)$. 
        \item If all possible $g$ values give uniform distributions of $c_i(1)-g ⋅ a_i(1)$, the $(a_i, c_i)$ are taken from the \textit{uniform distribution}. 
        \item If several $g$ appear to work, repeat the algorithm with more samples. 
    \end{enumerate}
\end{enumerate}

This attack works with probability $>\frac{1}{2}$ \cite{RLWE for NT}.
Similar attacks exist for $\gamma$ of small order, where $\gamma^r≡ 1$ for some small $r\in ℤ^+$.

Other attacks include exploitation of the size of the error values. However, the probability of success for this particular attack decays (except for in a special case) and is unlikely to be implemented \cite{provably weak}.


\begin{definition}[\cite{RLWE for NT}]
Let $f(x)\in ℤ[x]$ be a monic, irreducible polynomial of degree $n$ and let $\gamma$ be a root of $f(x)$. Then a \emph{smearing map} $\pi_{\gamma}$ is defined as
\begin{align*}
π_{\gamma}:P_q&\rightarrow \mathbb{Z}_q\\
    g(x)&\mapsto g(\gamma)
\end{align*}
\end{definition}

\begin{definition}[\cite{RLWE for NT}]
Given a smearing map $π_{\gamma}$ and a subset $S\subset P_q$, we say that $S$ \emph{smears} under $π_{\gamma}$ if $π_{\gamma}(\mathcal{S})=\mathbb{Z}_q$.
\end{definition}

Note that for a smearing map $π_{\gamma}$ we have $\ker \pi_{\gamma}=(x-\gamma)$.
Also, note that $\frac{P_q}{(x-\gamma)} \cong \Z_q$ where $P_q$ is a polynomial ring and $\gamma$ is one of its roots. This implies that $(x-\gamma)$ has $q$ cosets in $P_q$, which are, consequently, $\{i+(x-\gamma)\}_{i=0}^{q-1}$.

\begin{lemma} \label{images} 
Let $π_{\gamma}$ be a smearing map. Then, $π_{\gamma}(f)=π_{\gamma}(g)$ if and only if $f$ and $g$ are in the same coset of $(x-\gamma)$.
\end{lemma}

\begin{proof} The claim follows from the fact that
\begin{align*}
    f(\gamma)=g(\gamma)&\iff (f-g)(\gamma)=0\\
    &\iff f-g=h⋅ (x-\gamma) \text{ for some } h\in P_q\\
    &\iff f,g \text{ are in the same coset of }(x-\gamma).
\end{align*}
\end{proof}

This lemma implies that the set $S$ smears if and only if $S$ contains an element in each of the $q$ cosets of $(x-\gamma)$ in $P_q$.
In the next two sections we investigate the \emph{size} of a subset sampled from a \emph{uniform distribution} and investigate the properties of a subset sampled from a \emph{Gaussian distribution} as in the PLWE problem. 

\subsection{Smearing: The Uniform Distribution Case}
We investigate the \emph{size} of a subset sampled from a \emph{uniform distribution} i.e the polynomials in $S$ are chosen \emph{uniformly random} over $P_q$. 
As we will see, the assumption of uniformity eliminates much of the algebraic aspects of the smearing problem as related to PLWE and reduces the problem to a classic problem in probability theory, \emph {the Coupon Collector's Problem}. 

The classical version of the Coupon Collector’s Problem is as follows: Suppose a company places one of $q$ distinct types of coupons, $c_1, \ldots, c_p$, into each of its cereal boxes independently, with equal probability $\frac{1}{q}$. Let $X$ be a random variable indicating the number of objects one has to buy before collecting all of the coupons. The question then is:
\begin{center}
 How many boxes should one expect to buy before collecting at least one of each type of coupon? Equivalently, what is $\mathbb{E}[X]$?
\end{center}

The following is a well-known lemma which computes $\mathbb{E}[X]$ with a geometric distribution approach.
\begin{lemma}[\cite{ferrante}]
Let $X$ be the number of boxes needed to be purchased to collect all $q$ coupons. Then,
\[\mathbb{E}[X]=qH_q=q\log q+γ q + \frac{1}{2}+O(1/q)\]
where $H_q$ is the $q$-th harmonic number and $γ \approx 0.57722$ is the Euler-Mascheroni constant. 
Furthermore,
$$\text{Var}(X)<\frac{π^2}{6}q^2.$$
\end{lemma}


We can reduce the problem of uniform smearing to the Coupon Collector's Problem.
\begin{lemma}
A uniform distribution over $P_q$ maps under $\pi_{\gamma}$ to a uniform distribution over $\Z_q$.
\end{lemma}
\begin{proof}
By 
Lemma \ref{images}, and the fact that all cosets of $(x-\gamma)$ are of the same size, a polynomial chosen uniformly at random in $P_q$ will have probability $1/q$ of being in any given coset of $(x-\gamma)$ and hence there is a $1/q$ probability that $\pi_{\gamma}$ produces any given element in $\Z_q$. 

\end{proof}

So, instead of selecting polynomials in $P_q$ we can choose elements of $\Z_q$ uniformly. In this context, the smearing problem is identical to the Coupon Collector's Problem. Each polynomial has an image uniformly chosen between $1$ and $q$, and we want to ``collect all the coupons," i.e. for each element $j\in \Z_q$, collect at least one polynomial having $j=\pi_{\gamma}$ as its image under the smearing map.
\subsubsection{The Principal Question}
Let $m=|S|$ and $\Z_q$ be given. Assume that the elements of $S$ are chosen uniformly at random ($s\in S$, $s\leftarrow \mathcal{U}_{P_q}$). The question is to determine the probability that $S\subset P_q$ will smear. 
\begin{remark}
Note that, although in the smearing problem $q$ must be prime, in the broader context of the coupon collector's problem $q$ can be any positive integer; we will not demand such a restriction within our probabilistic calculations.
\end{remark}

Fix a polynomial $f(x)\in\mathbb{Z}[x]$ and thus fix some root $\gamma\in \mathbb{Z}_q$ and the smearing map $\pi_{\gamma}$. Denote with $P(m,q)$ the probability that a subset $S\subset P_q$ of size $m$ smears.

\begin{remark}
Given a probability distribution on $X$ (the random variable from the coupon collector's problem representing the number of cereal boxes one must purchase before collecting all $q$ coupons), $P(m,q)$ is simply the cumulative distribution function of $X$, since $P(S \text{ smears})=P(X\le |S|=m).$
\end{remark}

We compute and approximate smearing probabilities in Section 4.
\subsection{Smearing: The Non-Uniform Case}
In this section, we investigate the smearing condition when the error distribution over $P_q$ is not uniform. Note first that we can view drawing from $\mathcal{D}_{P, s,\sigma}$
\[(a, a\cdot s +e), a\leftarrow \mathcal{U}_{P_q}, s\in P_q, e\leftarrow G_{\sigma, P_q}\] as simply drawing 
\[(a,e),  a\leftarrow \mathcal{U}_{P_q},e\leftarrow G_{\sigma, P_q}\] since multiplying an element $a$ selected uniformly at random by a fixed secret $s$ is the same as selecting $a\cdot s$ uniformly at random, which, when the Gaussian distribution is added, yields the same Gaussian distribution for $e$. When we discuss the mapped error distribution, we consider selecting $e\leftarrow G_{\sigma, P_q}$ and its mapping $\pi_{\gamma}(e)$.

\subsubsection{The Distribution of $e(α)$}
An explicit method of calculating the probability distribution of $e(α)$ given the distribution of the polynomial coefficients of $e$ is presented in \cite{error analysis}.

\begin{theorem}\cite{error analysis}
Suppose $e_0, e_1, ..., e_{n-1}$ are independent random variables in $\Z_q$ with the same probability distribution $(c_0, c_1, ..., c_{q-1})$. Let $c(x)=∑_{i\in\Z_q}c_ix^i$. Then, for any $a_1, ..., a_{n-1}\in\Z_q$, the probability distribution of $e_0+e_1a_1+⋯+e_{n-1}a_{n-1}\mod q$ can be computed as the coefficients of the polynomial
$$c(x)c(x^{a_1})⋯ c(x^{a_{n-1}}) \mod x^q-1.$$
\end{theorem}

Since $e(α)=e_0+e_1α+e_2α^2+⋯+e_{n-1}α^{n-1}$,

by setting $a_i=α^i$, $1≤ i ≤ n-1$ and using the theorem above we can compute the probability distribution of $e(α)$ over $\Z_q$. In general, we refer to that distribution as $χ$, the ``mapped error distribution''.

We can compute the discrete Gaussian distribution over $ℤ_q = ℤ/qℤ$ as in \cite{error analysis}. The parameter used is $\sigma =\frac{\beta}{\sqrt{2\pi}}*q$.
\begin{figure}
    \centering
    \includegraphics[width=.75\textwidth]{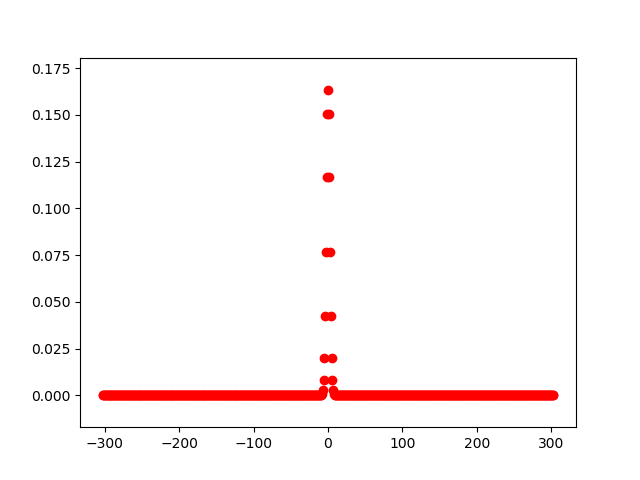}
    \caption{Initial Gaussian Distribution over $\mathbb{Z}_{607}$ with $\beta=.01$}
    \label{fig: initial}
\end{figure}
\begin{figure}
    \centering
    \includegraphics[width=.75\textwidth]{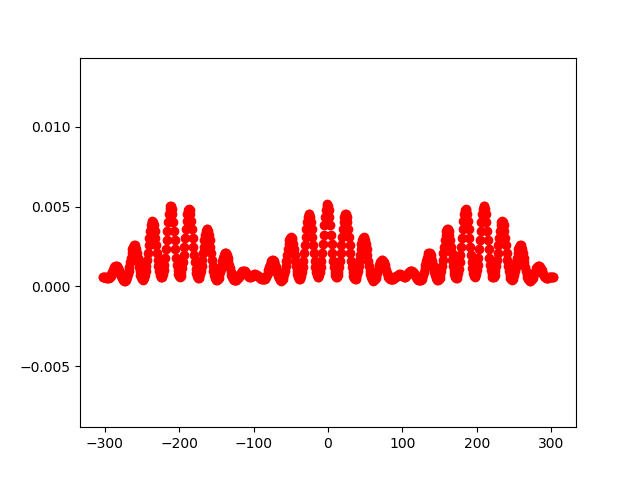}
    \caption{Mapped Error Distribution given Figure \ref{fig: initial} over $\mathbb{Z}_{607}$ with $\gamma = 396$, $\gamma^3\equiv 1 \mod 607$}
\end{figure}
Note that the low multiplicative order of $\gamma =396$ over $\mathbb{Z}_{607}$ gives the mapped error distribution structure; this illustrates the setting for the $\gamma$ of low order attack, as the mapped error distribution is certainly not uniform. 

\section{Computing Smearing Probabilities}

Let $\chi$ be a discrete probability distribution on $[q]$ (where $[q]$ denotes the set $\{1, 2,\ldots q\}$), and let $P_{\chi}(m,q) $ denote the probability that, when $m$ samples are independently drawn from $\chi$, they will ``smear,'' i.e. each element in $[q]$ will be chosen at least once. When $\chi$ is the uniform distribution, we denote this probability as $P_U(m,q)$, or simply as $P(m,q)$. In this section, we provide practical ways of calculating these probabilities.

\subsection{An Approximation of Uniform Smearing Probabilities}

A result by Erdős and Rényi \cite{erdos} gives a way to approximate $P(m,q)$ for large values of $q$.

\begin{theorem}[Erdős, Rényi \cite{erdos}]\label{erdos approx}

Let $U$ be the uniform distribution over $[q]$, and let $X$ be the random variable denoting the number of independent samples one must take from $U$ until picking each of $[q]$ at least once. Then,

\[\lim_{q\rightarrow\infty}\Pr(X<q\log q+cq)=\exp(-\exp(-c)).\]

\end{theorem}

In our case, $P(m,q) = \Pr(X\leq m)$, so making the substitution $m=q\log q+cq$ gives the formula

\[P(m,q)\approx\exp\left(-q \exp\left(-\frac{m}{q}\right)\right).\]

Although this is a powerful approximation, for some applications it might be preferable to calculate this probability exactly for concrete values of $m$ and $q$. The following sections contribute towards this goal, as well as give formulas for the case when $\chi$ is not uniform.

\subsection{A Recursive Formula in $m$}\label{m-rec-section}

\begin{proposition}\label{m-rec}

Let $\chi$ be a discrete probability distribution on $[q]$, with $p_k$ being the probability of picking the $k$th element. Let $\chi/k$ denote the probability distribution on $q-1$ elements after the $k$th element has been removed from $\chi$, and the remaining probabilities have been normalized. Then,

\[P_{\chi}(m,q)=P_{\chi}(m-1,q)+\sum_{k=1}^qp_k(1-p_k)^{m-1}\cdot P_{\chi/k}(m-1,q-1).\]

\end{proposition}

\begin{proof}

Assume that we choose $m$ independent samples one-by-one from the distribution $\chi$. Let $S$ be the event that the samples smear. Let $A$ be the event that the $m^{th}$ sample achieves smearing (i.e. the previous $m-1$ samples cover $q-1$ distinct elements, and the $m$th sample happens to cover the remaining $q$th element). Also, let $B$ be the event that smearing happens within the first $m-1$ samples (i.e. by the time $m-1$ samples have been taken, they already take on $q$ distinct values). Notice that $S = A \sqcup B$. Therefore, 
\[P_{\chi}(m,q)=\Pr(S)=\Pr(A)+\Pr(B).\]
To calculate $P(A)$, we use the Law of Total Probability to condition on the outcome of the $m^{th}$  sample, which we denote by $K\in[q]$:

\begin{align*}
    \Pr(A)&=\sum_{k=1}^q\Pr(A|K=k)\cdot \Pr(K=k)\\
    &=\sum_{k=1}^q\Pr(A|K=k)\cdot p_k.
\end{align*}

To calculate the value of $\Pr(A|K=k)$, we notice that the only way that smearing is achieved by the $m$th sample being equal to $k$ is if, first, the previous $m-1$ samples all fall into $[q]/k$, and, second, if the previous $m-1$ samples smear on $[q]/k$. Therefore, 

\[\Pr(A|K=k)=(1-p_k)^{m-1}\cdot P_{\chi/k}(m-1,q-1),\]

where $(1-p_k)^{m-1}$ is the probability that the first $m-1$ samples are contained in $[q]/k$, and $P_{\chi/k}(m-1,q-1)$ is the probability that, conditioned on this, these $m-1$ samples smear on $[q]/k$. Hence,

\[\Pr(A)=\sum_{k=1}^qp_k(1-p_k)^{m-1}\cdot P_{\chi/k}(m-1,q-1).\]

On the other hand, the probability of event $B$, or that smearing is achieved within the first $m-1$ samples, is simply $P_{\chi}(m-1,q)$, so 

\[P_{\chi}(m,q)=P_{\chi}(m-1,q)+\sum_{k=1}^qp_k(1-p_k)^{m-1}\cdot P_{\chi/k}(m-1,q-1).\]

\end{proof}

In the case where $\chi$ is the uniform distribution on $[q]$, this relation becomes greatly simplified:

\begin{lemma}\label{un-rec}

For the uniform distribution on $[q]$,

\[P(m,q)=P(m-1,q)+P(m-1,q-1)\cdot\left(\frac{q-1}{q}\right)^{m-1}.\]

\end{lemma}

\begin{proof}

We use the result of Proposition \ref{m-rec}. In the uniform distribution, $p_k=\frac{1}{q}$ for every $k$. Furthermore, for every $k$, $\chi/k$ is the uniform distribution on $[q-1]$, so $P_{\chi/k}(m-1,q-1)=P(m-1,q-1)$. Therefore, if $\chi$ is the uniform distribution, then

\begin{align*}
    P_{\chi}(m,q)&=P_{\chi}(m-1,q)+\sum_{k=1}^qp_k(1-p_k)^{m-1}\cdot P_{\chi/k}(m-1,q-1)\\
    &=P(m-1,q)+\sum_{k=1}^q\frac{1}{q}\left(\frac{q-1}{q}\right)^{m-1}\cdot P(m-1,q-1)\\
    &=P(m-1,q)+P(m-1,q-1)\cdot\left(\frac{q-1}{q}\right)^{m-1}.
\end{align*}
\end{proof}


Lemma \ref{un-rec}, if implemented as a recursive formula, provides a very rapid method of computing $P(m,q)$. The base cases are rather straightforward. If $m<q$, then
\[P(m,q)=0,\]
since it is impossible to pick $q$ different elements with fewer than $q$ samples. If $m=q$, then 
\[P(m,q)=\frac{q!}{q^q}.\] 

To see why this is the case, notice that if the number of samples is equal to $q$, then every single sample must be a ``success,'' i.e. pick an element of $[q]$ that has not been picked before. For the first sample, one can pick any of $[q]$, so the probability of success is $q/q$. For the second sample, there are $q-1$ unpicked elements, so the probability of success is $(q-1)/q$. For the third sample, there are now $q-2$ elements that are not selected, so the probability of success is now $(q-2)/q$. This continues until the $q^{th}$ sample, for which there is only one option left, giving a success probability of $1/q$. Multiplying these probabilities together gives $q!/q^q$. Finally, if $q=1$, and $m>0$, then $P(m,q)=1$, as the first sample is always a success, and one success is sufficient in this case.

Computing $P(m,q)$ using the recursive formula of Lemma \ref{un-rec} along with these base cases results in the computation of $P(i,j)$ for each of $1\leq j\leq q$ and $j\leq i\leq j+(m-q)$. Hence, the complexity of the recursive computation is on the order of $q(m-q)$. Notice, however, that as a result, one computes not just $P(m,q)$, but also $P(i,q)$ for all $i\in[0,m]$, which is very useful information to have for choosing parameters for the smearing attack (which will be discussed later).

\begin{figure}
\includegraphics[width=.75\textwidth]{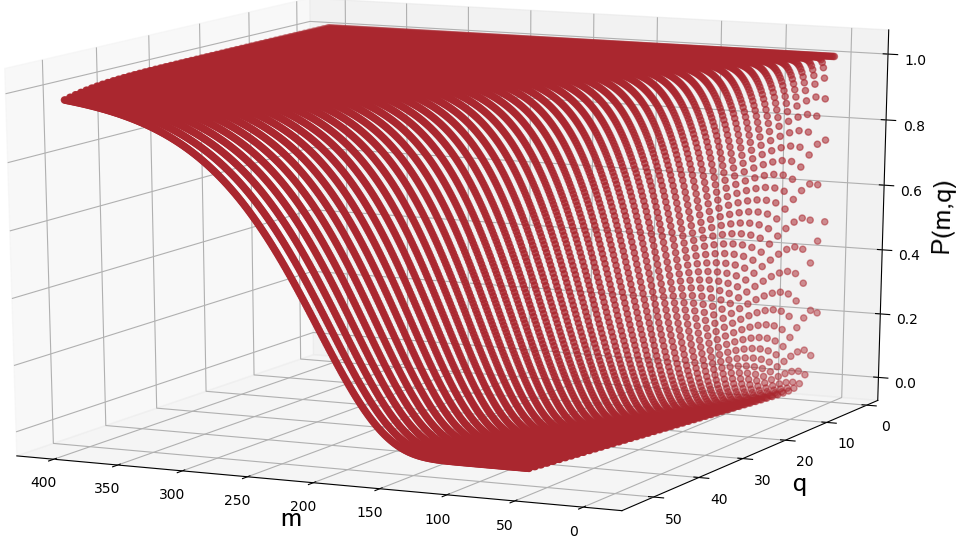}
\caption{$P(m,q)$, $1 \le m\le 400$ and $1\le q\le 53$.}
\label{un-graph}
\end{figure}

Figure \ref{un-graph} shows $P(m,q)$ values for a range of $m$ and $q$, calculated using this method.

While Lemma \ref{un-rec} provides an effective method of calculating smearing probabilities for uniform distributions, using Proposition \ref{m-rec} to calculate smearing probabilities for non-uniform distributions is inefficient, since, to calculate $P_{\chi}(m,q)$, one must calculate smearing probabilities on all subsets of $[q]$, which makes the complexity on the order of $2^q\cdot (m-q)$. A more efficient method for non-uniform smearing can be achieved by recursion in $q$, rather than recursion in $m$, as described in the next section.

\subsection{A Recursive Formula in $q$.}
\begin{proposition}\label{q-rec}
Let $\chi$ be a discrete probability distribution on $[q]$, with $p_q$ being the probability of picking the $q^{th}$ element. Let $\chi/q$ denote the probability distribution on $q-1$ elements after the $q^{th}$ element has been removed from $\chi$, and the remaining probabilities have been normalized. Then,
\[P_{\chi}(m,q)=\sum_{k=1}^{m-q+1}\binom{m}{k}p_q^k(1-p_q)^{m-k}\cdot P_{\chi/q}(m-k,q-1).\]
\end{proposition}

\begin{proof}
Let $K$ be a random variable denoting the number of times the $q$th element is picked. For smearing to occur, $K$ must be at least 1 (else the $q^{th}$ element will not be chosen), but cannot be greater than $m-q+1$. This is because there are $m$ samples in total, and one needs at least $q-1$ of them to cover the first $q-1$ elements, leaving a maximum of $m-(q-1)$ available for the $q$th element. Then, by the Law of Total Probability,

\[P_{\chi}(m,q)=\sum_{k=1}^{m-q+1}\Pr(K=k)\cdot\Pr(\text{smearing}|K=k).\]

Since the samples are drawn independently, notice that $K\sim\text{Bin}(m, p_q)$. Therefore,

\[\Pr(K=k)=\binom{m}{k}p_q^k(1-p_q)^{m-k}.\]

On the other hand, the probability that smearing occurs given that the $q^{th}$ element is chosen $k$ times, where $k \in [1,m-q+1])$, is the probability that $m-k$ samples, taken from $\chi/q$, smear on $[q-1]$. 
Hence,

\[\Pr(\text{smearing}|K=k)=P_{\chi/q}(m-k,q-1).\]

Finally, \[P_{\chi}(m,q)=\sum_{k=1}^{m-q+1}\binom{m}{k}p_q^k(1-p_q)^{m-k}\cdot P_{\chi/q}(m-k,q-1).\]

\end{proof}

Using Proposition \ref{q-rec} as a basis for a recursive method of calculating $P_{\chi}(m,q)$ is more efficient than using Proposition \ref{m-rec}. The base cases are very similar to those in Section \ref{m-rec-section}. If $m<q$, then $P_{\chi}(m,q)=0$, and if $q=1$ and $m>0$, then $P_{\chi}(m,q)=1,$ for the same reasons as in the uniform case. To find $P_{\chi}(q,q)$, notice that, as in the uniform case, each sample must pick an element of $[q]$ which has not been picked before. Hence, if $q$ samples smear, they must be a permutation of $[q]$. Each such permutation has a probability of $\prod_{k=1}^q p_k$, where $p_k$ is the probability of picking the $k$th element, and there are $q!$ such permutations, meaning the probability of smearing is
\[P_{\chi}(q,q)=q!\cdot\prod_{k=1}^q p_k.\]
As expected, when $\chi$ is uniform, $p_k=1/q$ for every $k$, so the formula simplifies to the one in Section \ref{m-rec-section}.

Calculating $P_{\chi}(m,q)$ recursively using Proposition \ref{q-rec} along with these base cases results in the computation of $P_{\chi/[j+1,q]}(i,j)$ for each of $1\le j\le q$ and $j\le i\le j+(m-q)$. In turn, the computation of each of these values requires a sum of on the order of $m-q$ terms. Hence, the complexity of this recursive method is on the order of $q\cdot(m-q)^2$. Notice that, as in the recursion-in-$m$ method, this recursion-in-$q$ method results in the calculation of not just $P(m,q)$, but also $P_{\chi}(i,q)$ for all $i\in[0,m]$. Of course, an attacker on PLWE would not have prior knowledge of the non-uniform distribution $\chi$, but such information is nevertheless useful in a retrospective analysis of the effectiveness of the smearing attack (discussed later).

\section{The Smearing Attack}

Here, we build on the previous sections to present a smearing-based attack on Polynomial Learning With Errors, which we call the ``smearing attack.'' 

\subsection{The Uniform Distribution Smears the Best}

A fundamental principle regarding uniform and non-uniform smearing is as follows. We begin with the following lemma.

\begin{lemma}\label{averaging}

Let $\chi$ be a distribution on $[q]$, and let $i\neq j$, be two elements of $[q]$. Let $p_i$ and $p_j$ be the probabilities of selecting $i$ and $j$ respectively. Construct $\chi'$ as follows: take $\chi$, and replace the probabilities of $i$ and $j$ with $(p_i+p_j)/2$. Then,

\[P_{\chi}(m,q)\leq P_{\chi'}(m,q),\]

with equality if and only if $p_i=p_j$.

\end{lemma}

\begin{proof}

Define $K$ as the random variable representing the number of samples from the distribution which fall into $\{i,j\}$. Notice that $K\sim \mbox{Bin}(m,p_i+p_j)$, for both $\chi$ and $\chi'$. Let $k$ be a specific instance of $K$, $k\ge 2$. Conditioned on $K=k$, smearing on $\{i,j\}$ is independent from smearing on $[q]/\{i,j\}$, and since the probabilities on $[q]/\{i,j\}$ are unchanged between $\chi$ and $\chi'$, to compare $P_{\chi}(m,q)$ and $P_{\chi'}(m,q)$ it suffices to compare the probabilities of picking at least one of each of $i$ and $j$ for both distributions, restricted to what is going on with these $k$ samples. For $\chi$, this probability is

\begin{align*}
    \Pr(\text{picking both }i\text{ and }j|K=k) 
    &= 1-\Pr(\text{not picking both }i\text{ and }j|K=k)\\
    &= 1-(\Pr(\text{picking only }i|K=k)+\Pr(\text{picking only }j|K=k))\\
    &= 1-\left(\left(\frac{p_i}{p_i+p_j}\right)^k+\left(\frac{p_j}{p_i+p_j}\right)^k\right)\\
    &= 1-\frac{p_i^k+p_j^k}{(p_i+p_j)^k},
\end{align*}

since the probability of picking $i$ out of $\{i,j\}$ is $p_i/(p_i+p_j)$ and similarly with $j$. For $\chi'$, a similar computation shows that

\begin{align*}
    \Pr(\text{picking both }i\text{ and }j|K=k)
    &= 1-\frac{\left(\frac{p_i+p_j}{2}\right)^k + \left(\frac{p_i+p_j}{2}\right)^k}{\left(\left(\frac{p_i+p_j}{2}\right) + \left(\frac{p_i+p_j}{2}\right)\right)^k}\\
    &= 1-\frac{\left(\frac{p_i+p_j}{2}\right)^k + \left(\frac{p_i+p_j}{2}\right)^k}{(p_i+p_j)^k}.
\end{align*}

It remains, thus, to show that 

\[p_i^k+p_j^k\geq \left(\frac{p_i+p_j}{2}\right)^k+\left(\frac{p_i+p_j}{2}\right)^k,\]

with equality if and only if $p_i=p_j$. To show this, consider this as an optimization problem, where we try to minimize the quantity $f(P_i,P_j)=P_i^k+P_j^k$ subject to the constraint $g(P_i,P_j)=p_i+p_j$. By method of Lagrangian multipliers, setting $\nabla f=\lambda\nabla g$ gives

\[kP_i^{k-1}=kP_j^{k-1}=\lambda,\]

which implies that $P_i^k+P_j^k$ is minimized at $P_i=P_j=(p_i+p_j)/2.$ Hence,

\[p_i^k+p_j^k\geq\left(\frac{p_i+p_j}{2}\right)^k+\left(\frac{p_i+p_j}{2}\right)^k,\]

with equality if and only if $p_i=p_j$. This implies that

\[P_{\chi}(m,q|K=k)\leq P_{\chi'}(m,q|K=k)\]

for $k\ge 2$. Then, by the Law of Total Probability, 

\begin{align*}
    P_{\chi'}(m,q)-P_{\chi}(m,q)
    &= \left(\sum_{k=2}^{m-q+2} P_{\chi'}(m,q|K=k)\cdot\Pr(K=k)\right)
    - \left(\sum_{k=2}^{m-q+2} P_{\chi}(m,q|K=k)\cdot\Pr(K=k)\right)\\
    &= \sum_{k=2}^{m-q+2} (P_{\chi'}(m,q|K=k)-P_{\chi}(m,q|K=k))\cdot\Pr(K=k)\\
    &\ge \sum_{k=2}^{m-q+2} 0\cdot\Pr(K=k)\\
    &= 0.
\end{align*}
Therefore,
\[P_{\chi}(m,q)\leq P_{\chi'}(m,q),\]
with equality if and only if $p_i=p_j$, as seen from the optimization problem.
\end{proof}


\begin{theorem}
Let $\chi$ be a probability distribution over $[q]$, and let $U$ be the uniform distribution over $[q]$. Then, 
\[P_{\chi}(m,q)\leq P_U(m,q),\]
with equality if and only if $\chi=U$.
\end{theorem}
\begin{proof}
From $\chi=\chi_0$, build a new distribution $\chi_{N+1}$ by selecting two elements in $[q]$ from the previous distribution $\chi_N$, and replacing their two probabilities $p_i$ and $p_j$ with $(p_i+p_j)/2$, their average. By Lemma \ref{averaging}, $$P_{\chi_{N+1}}(m,q)\geq P_{\chi_N}(m,q)$$
We construct a sequence $\{\chi_N\}$ such that $\{P_{\chi_N}(m,q)\}_{N=0}^{\infty}$ is a non-decreasing, infinite sequence with limit $P_U(m,q)$. This shows that $P_{\chi}(m,q)\le P_U(m,q)$. 
Furthermore, if $P_{\chi}(m,q)=P_U(m,q)$, then the sequence is constant, meaning $P_{\chi_N}(m,q)=P_{\chi_{N+1}}(m,q)$ for all $N$. By Lemma \ref{averaging}, this is only possible if, for each step, $p_i=p_j$, meaning that $\chi_0=\chi_1=...=\chi_N=....$ Then, $\{\chi_N\}_{N=0}^{\infty}$ is a constant sequence with limit $U$, hence $\chi=U$.  
\end{proof}
This principle is the driving force behind the attack on the Decision-PLWE problem, as described in the following sections.
\subsection{The Smearing Decision Problem}
The foundation of the smearing attack is in what we call the ``smearing decision": given a large number of samples from some probability distribution over $\Z_q$, decide, with some certainty, whether that distribution is the uniform distribution $U$, or a certain non-uniform distribution $\chi$. We do this in the following way.
\begin{enumerate}
\item Choose the parameters $N$, indicating the number of trials to be done, and $m$, the number of samples to be taken per trial. $N$ must be odd, while $m$ must be picked such that $P_U(m,q) > 1/2$ while $P_{\chi}(m,q) < 1/2$. Since $P_U(m,q) > P_{\chi}(m,q)$ for all $m\ge q$, and both $P_U(m,q)$ and $P_{\chi}(m,q)$, as functions of $m$, have range $[0,1)$, such an $m$ exists almost always. 

\item For each trial, take $m$ samples, and check whether they smear on $\Z_q$. If smearing happens for more than half of the $N$ trials, conclude that the samples were taken from the uniform distribution over $\Z_q$. If, on the other hand, the smearing happens for less than half of the $N$ trials, conclude that the samples were taken from $\chi$. 
\end{enumerate}
\begin{figure}
        \includegraphics[width=0.6\textwidth]{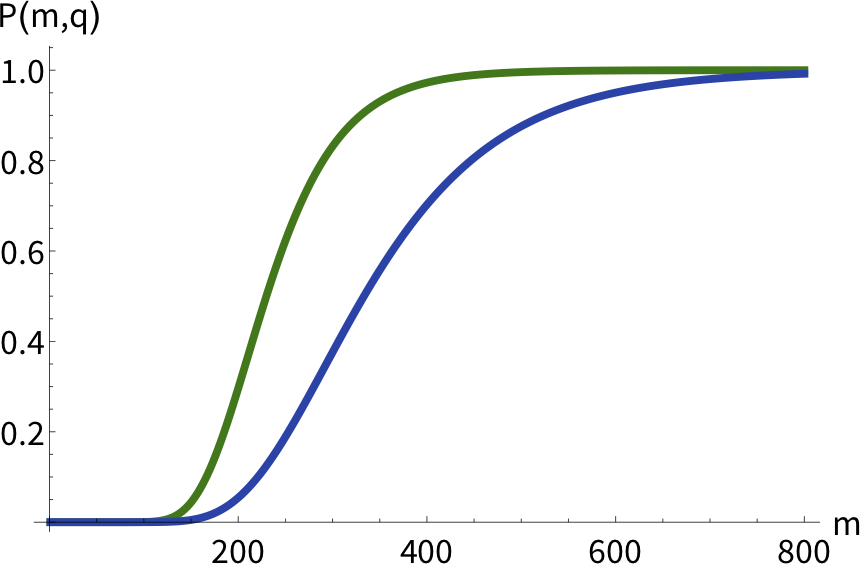} 
        \caption{Probability of smearing for a uniform distribution (green) and a non-uniform distribution (blue), as a function of $m$ ($q=53$, $n=2$, $\sigma=6$, $\gamma=2$).}
        \label{un-vs-chi}
\end{figure}  

\begin{figure}  
        \includegraphics[width=0.6\textwidth]{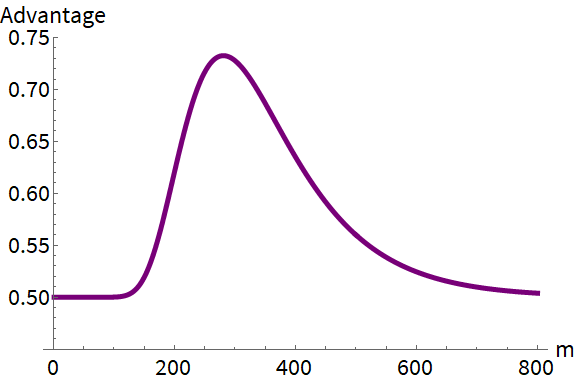}
        \caption{The probability of deciding between them correctly by whether smearing happens or not.}
        \label{advantage}
\end{figure}

To give an intuitive explanation of this decision process, consider the two graphs shown in Figures \ref{un-vs-chi} and Figure \ref{advantage}. Figure \ref{un-vs-chi} shows the smearing probability for the uniform distribution (in green), and a non-uniform distribution (in blue) as a function of $m$. The non-uniform distribution is a mapped Gaussian distribution, with the parameters $q=53$, $n=2$, $\sigma=6$ (the standard deviation of the initial Gaussian), and $\gamma=2$. As expected, for both curves, when the number of samples is small, the probability of smearing is $0$ (or close to $0$), while when the number of samples if large, smearing occurs almost always. The uniform and non-uniform curves can really be differentiated only for some intermediate range of $m$. 

We describe here a simple example of the smearing decision and the smearing attack in general. Suppose that $U$ or $\chi$ have equal probability.
Assume that if smearing happens, the distribution is assumed to be uniform, and if smearing does not happen, the distribution is assumed to be non-uniform. Then, the 
the probability that the decision is correct, is
\begin{align*}
    \Pr(\text{decision is correct}) 
    &= \Pr(\text{decision is correct}|U)\Pr(U) 
    + \Pr(\text{decision is correct}|\chi)\Pr(\chi)\\
    &= \Pr(\text{smearing happens}|U)\cdot\frac{1}{2}
    + \Pr(\text{smearing doesn't happen}|\chi)\cdot\frac{1}{2}\\
    &= \frac{1}{2}(P_U(m,q)+(1-P_{\chi}(m,q))) \\
    &= \frac{1}{2}+\frac{1}{2}(P_U(m,q)-P_{\chi}(m,q)).
\end{align*}
Notice that since $P_U(m,q)>P_{\chi}(m,q)$, the 
probability is strictly greater than half, and increases linearly with the difference in the smearing probabilities of the uniform and non-uniform distributions. The graph of the probabilities 
for different values of $m$ is shown in Figure \ref{advantage}. As expected, the probability is highest when the distance between the two smearing probability curves is the greatest. 


The following proposition formalizes the idea of the smearing decision.
\begin{proposition}\label{N}
Let $U$ be the uniform distribution over $\Z_q$, and $\chi$ be some non-uniform distribution over $\Z_q$. Let $m$ be an integer such that $P_U(m,q) > 1/2$ and $P_{\chi}(m,q) < 1/2$. Then, given arbitrarily small $\alpha, \beta > 0$, there exists an $N$ such that the smearing decision with $N$ trials is correct, in the case where the true distribution is $U$, with probability $1-\alpha$, and in the case where the true distribution is $\chi$, with probability $1-\beta$.
\end{proposition}
\begin{proof}
Consider the case in which the unknown distribution about which the decision is being made is actually $U$. Define $X$ to be a random variable denoting the number of trials for which the samples smear. In this case, for each trial, smearing happens with probability $P_U(m,q)$ (which we denote as simply $P_U$ for convenience), and the trials are independent from one another. Hence, $X \sim \text{Bin}(N, P_U)$, so $\E(X)=NP_U$ and $\Var(X)=NP_U(1-P_U)$. In this case, the probability that the smearing decision is incorrect is the probability that fewer than $N/2$ trials smear. Using Chebyshev's Inequality, which states that for a random variable $X$, and any $c>0$,
\[\Pr(|X-\E(X)|\ge c\sigma)\le\frac{1}{c^2},\]
we conclude that 
\begin{align*}
    \Pr(\text{decision is incorrect}|U) &= \Pr\left(X < \frac{N}{2}\right)\\
    &= \Pr\left(NP_U-X > N\left(P_U-\frac{1}{2}\right)\right)\\
    &\le \Pr\left(|NP_U-X| \ge N\left(P_U-\frac{1}{2}\right)\right)\\
    &\le \frac{P_U(1-P_U)}{N\left(P_U-\frac{1}{2}\right)^2}.
\end{align*}
Hence,
\[\lim_{N\rightarrow\infty}\Pr(\text{decision is incorrect}|U)=0,\]
so, in particular, we can choose  $N_1$ large enough such that this quantity is less than $\alpha$. On the other hand, in the case where the unknown distribution is $\chi$, $X\sim\text{Bin}(N,P_{\chi})$ (where we denote $P_{\chi}(m,q)$ by $P_{\chi}$ for convenience), the decision is incorrect whenever smearing happens in more than $N/2$ trials. By a similar argument as above,
\begin{align*}
    \Pr(\text{decision is incorrect}|\chi) &= \Pr\left(X > \frac{N}{2}\right)\\
    &= \Pr\left(X-NP_{\chi}>N\left(\frac{1}{2}-P_{\chi}\right)\right)\\
    &\le \Pr\left(|X-NP_{\chi}|\ge N\left(\frac{1}{2}-P_{\chi}\right)\right)\\
    &\le \frac{P_{\chi}(1-P_{\chi})}{N\left(\frac{1}{2}-P_{\chi}\right)^2}.
\end{align*}
Thus, 
\[\lim_{N\rightarrow\infty}\Pr(\text{decision is incorrect}|\chi)=0,\]
In particular, choose $N_2$ large enough such that this quantity is less than $\beta$. Lastly, choose $N=\max(N_1, N_2)$, which then satisfies both conditions of the proposition.
\end{proof}


\subsection{The Smearing Attack}
The proposed smearing attack on the Decision-PLWE problem proceeds as follows. Assume that we know $P_q$ (and a root $\gamma$ of the polynomial $f(x)$), and have access to a large number of samples $(a_i, b_i) \in P_q \times P_q$. 

{\bf Algorithm.} The algorithm for the smearing attack works as follows:
\vspace{-0.1in}
\begin{enumerate}
\item Choose the parameters $m$ and $N$ to achieve the desired error probabilities $\alpha$ and $\beta$.
\item For all possible guesses $g\in \Z_q$ for the true value of $s(\gamma)$, make the following smearing decision: for each of $N$ trials, draw $m$ samples $(a_i, b_i)$, compute the set $S= \{b_i(\gamma)-a_i(\gamma)\cdot g\}$, and determine whether $S$ smears, i.e. $S=\Z_q$.
\begin{enumerate}
    \item If smearing occurs in more than half of the trials ($>N/2$ trials), conclude that the error distribution resulting from the guess $g$ is uniform.
    \item Else conclude that the error distribution from the guess $g$ is non-uniform (in fact, it is the mapped Gaussian distribution). 
\end{enumerate}
\item Make a decision about the sample distribution:
\begin{enumerate}
 \item If the error distribution is uniform for all values of $g\in\Z_q$, conclude that the samples $(a_i, b_i)$ originally came from a uniform distribution.
 \item Else, if the error distribution is uniform for all but one values of $g$, conclude that the samples $(a_i, b_i)$ originally came from the PLWE distribution. In this case, the value of $g$ which gives a non-uniform error distribution is likely to be the true value of the secret $s$ evaluated at $\gamma$. 
 \item If the error distribution is non-uniform for more than one value of $g$, choose better values for the parameters $m$ and $N$ and repeat the steps of the algorithm.
\end{enumerate}
\end{enumerate}

If $(a_i, b_i)$ actually came from the uniform distribution over $P_q\times P_q$, then, for any value of $g\in\Z_q$, the values of $b_i(\gamma)-a_i(\gamma)\cdot g$ would also be uniformly distributed over $\Z_q$. If, on the other hand, $(a_i, b_i)$ came from the PLWE distribution, but $g$ was an incorrect guess for $s(\gamma)$, then the $b_i(\gamma)-a_i(\gamma)\cdot g$ values would also be uniformly distributed over $\Z_q$. Finally, if $(a_i, b_i)$ came from the PLWE distribution, and $g$ was the correct guess for $s(\gamma)$, only in this case would the $b_i(\gamma)-a_i(\gamma)\cdot g$ values follow the true distribution of the PLWE error term $e$ when mapped to $\Z_q$ by evaluating $e$ at $x=\gamma$.

For well-chosen parameters $m$ and $N$, the smearing attack described above correctly distinguishes between the PLWE and the uniform distributions over $P_q\times P_q$ almost always, as formalized the following proposition.

\begin{proposition}
Let $\alpha$ and $\beta$ be respectively the Type 1 and Type 2 errors of the smearing decision over $\Z_q$. Then, if the true distribution of the samples $(a_i, b_i)$ is the uniform distribution over $P_q\times P_q$, the smearing attack gives the correct decision with probability
\[\frac{1-\alpha}{1+(q-1)\alpha}.\]
If the true distribution of the samples $(a_i, b_i)$ is the PLWE distribution over $P_q\times P_q$, the smearing attack gives the correct decision with probability
\[\frac{1-\alpha-\beta+q\alpha\beta}{1-\alpha+(q-1)\alpha\beta}.\]
\end{proposition}
\begin{proof}
First, assume that the original distribution over $P_q\times P_q$ is uniform. Then, we are concerned with only two outcomes: that all $q$ decisions indicate uniform distribution, or that one of the $q$ decisions indicates non-uniform distribution, while the rest indicate uniform. 

In this case, the probability that all $q$ decisions indicate that the error distribution is uniform is the probability that all the smearing decisions are correct, which is $(1-\alpha)^q$. On the other hand, the probability of one of the decisions indicating that the error distribution is non-uniform, while the rest indicating uniform is the probability that one of the smearing decisions is incorrect, while the rest are correct, which is $q\alpha(1-\alpha)^{q-1}$. In the first outcome, the attack would indicate that the distribution is uniform, while in the second outcome, the attack would indicate that the distribution is the PLWE distribution. Thus, the probability that in this case the attack is successful is

\[\frac{(1-\alpha)^q}{(1-\alpha)^q+q\alpha(1-\alpha)^{q-1}}
=\frac{1-\alpha}{1+(q-1)\alpha}.\]

On the other hand, assume that the original distribution over $P_q\times P_q$ is the PLWE distribution. In this case, for $q-1$ of the guesses for $g$, the error distribution will be uniform, while for the correct guess, the error distribution will be the mapped Gaussian distribution. Correspondingly, there are three possible outcomes:
\vspace{0.1in}
\begin{enumerate}
    \item All $q$ of the smearing decisions could indicate that the error distribution over $\Z_q$ is uniform, which means that $q-1$ of the smearing decisions were correct with choosing the uniform distribution, while the smearing decision corresponding to the correct guess was incorrect in choosing the uniform distribution. The probability of this outcome is $(1-\alpha)^{q-1}\beta$.
    \item The $q-1$ of the smearing decisions could indicate that the error distribution is uniform, while the smearing decision corresponding to the correct guess could indicate that the error distribution is non-uniform. This means all smearing decisions were done correctly, the probability of which is $(1-\alpha)^{q-1}(1-\beta)$.
    \item The $q-1$ of the smearing decisions could indicate that the error distribution is uniform, while one smearing decision could indicate that the error distribution is non-uniform, but that smearing decision is the one which corresponds to a wrong guess $g$. In this case, $q-2$ of the smearing decisions need to identify a uniform distribution correctly, the one corresponding to the right guess $g$ needs to fail in identifying a non-uniform distribution, and one smearing decision corresponding to a wrong guess for $g$ needs to fail in identifying a uniform distribution. The probability of this happening is $(q-1)(1-\alpha)^{q-2}\alpha\beta$.
\end{enumerate}
The first outcome gives an incorrect smearing decision, while the second and third would give the correct smearing decision. Therefore, the probability that the attack is successful if the original distribution is the PLWE distribution is 
\[
\frac{(1-\alpha)^{q-1}(1-\beta)+(q-1)(1-\alpha)^{q-1}\alpha\beta}{(1-\alpha)^{q-1}\beta+(1-\alpha)^{q-1}(1-\beta)+(q-1)(1-\alpha)^{q-2}\alpha\beta}
= \frac{1-\alpha-\beta+q\alpha\beta}{1-\alpha+(q-1)\alpha\beta}.\]
\end{proof}

\section{Conclusion}
In this work, we characterize the probability of smearing, and developing a smearing-base attack on the Decision PLWE problem. By extension, our analysis also bears influence on the security of the RLWE problems through the RLWE to PLWE reduction \cite{RLWE for NT} (for a specific class of rings). As an extension of this work, we may consider adapting our attack to RLWE and investigating its practical implementation.

\section*{Acknowledgements}
This research was supported by the National Science Foundation under the grant DMS-1659872.

\end{document}